\documentclass[a4paper,reqno]{amsart}
\usepackage{amsmath,amssymb,mathrsfs}
\usepackage{txfonts,bm,pifont}
\usepackage{cite}
\usepackage{float}
\usepackage{stackrel}
\usepackage{subcaption}
\usepackage{graphicx,xcolor}
\usepackage[dvipdfmx,colorlinks=true]{hyperref}
\usepackage{comment}
\usepackage{enumitem}
\usepackage{longtable}

\newtheorem{theorem}{Theorem}[section]
\newtheorem{proposition}[theorem]{Proposition}

\newtheorem{lemma}[theorem]{Lemma}
\newtheorem{remark}[theorem]{Remark}
\theoremstyle{definition}
\newtheorem{definition}[theorem]{Definition}

\numberwithin{equation}{section}
\numberwithin{figure}{section}
\numberwithin{table}{section}
\newcommand{\wutilde}[1]{\vrule depth 0pt width 0pt%
{\raise0.8pt\hbox{$\smash{{\mathop{#1} \limits_{\displaystyle\widetilde{}}}}$}}}
\newcommand{\wuhat}[1]{\vrule depth 0pt width 0pt%
{\raise0.6pt\hbox{$\smash{{\mathop{#1} \limits_{\displaystyle\widehat{}}}}$}}}
\newcommand{\ol}[1]{\overline{#1}}

\newcommand{\ul}[1]{\underline{#1}}

\newcommand{\al}{\alpha}

\newcommand{\de}{\delta}
\newcommand{\ga}{\gamma}

\newcommand{\ep}{\bm{\epsilon}}

\newcommand{\PDE}{P$\Delta$E}
\newcommand{\bbZ}{\mathbb{Z}}
\newcommand{\bbR}{\mathbb{R}}
\newcommand{\bbC}{\mathbb{C}}

\newcommand{\bml}{{\bm l}}

\newcommand{\set}[2]{\left\{\left. #1 ~\right|~ #2 \right\}}

\newcommand{\tW}{\widetilde{W}}

\newcommand{\bi}{{\overline{i}}}

\makeatletter
\long\def\@makecaption#1#2{
 \vskip 10pt
 \setbox\@tempboxa\hbox{#1. #2}
 \ifdim \wd\@tempboxa >\hsize #1. #2\par \else \hbox
to\hsize{\hfil\box\@tempboxa\hfil}
 \fi}
\makeatother
\begin{document}

\title[]{Reduction of quad-equations consistent around a cuboctahedron I: additive case}
\author{Nalini Joshi}
\address{School of Mathematics and Statistics F07, The University of Sydney, New South Wales 2006, Australia.}
\email{nalini.joshi@sydney.edu.au}
\author{Nobutaka Nakazono}
\address{Institute of Engineering, Tokyo University of Agriculture and Technology, 2-24-16 Nakacho Koganei, Tokyo 184-8588, Japan.}
\email{nakazono@go.tuat.ac.jp}
\begin{abstract}
In this paper, we consider a reduction of a new system of partial difference equations, which was obtained in our previous paper\cite{JN2019:arxiv1906.06650} and shown to be consistent around a cuboctahedron. We show that this system reduces to $A_2^{(1)\ast}$-type discrete Painlev\'e equations by considering a periodic reduction of a three-dimensional lattice constructed from overlapping cuboctahedra.
\end{abstract}

\subjclass[2010]{
14H70, 
33E30, 
34M55,
37K10, 
39A14, 
39A23, 
39A45 
}
\keywords{
Consistency around a cuboctahedron;
Consistency around an octahedron;
quad-equation;
Consistency around a cube;
ABS equation;
Discrete Painlev\'e equation
}

\maketitle
\allowdisplaybreaks
\section{Introduction}\label{Introduction}
In this paper, we consider a system of partial difference equations governing a function $u=u(\bml)$ taking values on the vertices of a face-centered cubic lattice $\Omega$, given by
\begin{equation}\label{eqn:def_Omega}
 \Omega=\set{\bml=\sum_{i=1}^3l_i\ep_i}{l_i\in\bbZ,~l_1+l_2+l_3\in2\bbZ},
\end{equation}
where $\{\ep_1,\ep_2,\ep_3\}$ is a standard basis of $\bbR^3$. The system consists of 6 equations:
\begin{equation}\label{eqn:PDEs_P123456}
 \dfrac{u_{\ol{ik}}}{u_{\ul{i}\ol{k}}}
 =\dfrac{(\al_{ij}+\ga_i)u_{\ol{jk}}-(\al_{ij}+\ga_j-\ga_k)u_{\ul{j}\ol{k}}}
 {(\al_{ij}-\ga_j+\ga_k)u_{\ol{jk}}-(\al_{ij}-\ga_i)u_{\ul{j}\ol{k}}},\quad
 \dfrac{u_{\ul{jk}}}{u_{\ol{jk}}}
 =\dfrac{(\al_{ij}+\ga_i)u_{\ul{ik}}-(\al_{ij}-\ga_j+\ga_k)u_{\ol{ik}}}
 {(\al_{ij}+\ga_j-\ga_k)u_{\ul{ik}}-(\al_{ij}-\ga_i)u_{\ol{ik}}},
\end{equation}
where $(i,j,k)=(1,2,3),\,(2,3,1),\,(3,1,2)$, and the bars $\bar{i}$ and $\ul{j}$ denote $\bml\to\bml+\ep_i$ and $\bml\to\bml-\ep_j$ respectively and the coefficients are given by
\begin{subequations}
\begin{align}
 &\al_{ij}=\al_i(l_i)-\al_j(l_j),
 &&\al_i(k)=\al_i(0)+k,
 &&i,j\in\{1,2,3\},~k\in\mathbb{Z},\\
 &\ga_1=-c+(-1)^{l_1+l_2}\de_1,
 &&\ga_2=-c+(-1)^{l_2+l_3}\de_2,
 &&\ga_3=-c+(-1)^{l_1+l_3}\de_3,
\end{align}
\end{subequations}
with $\al_i(0)$, $i=1,2,3$, $c$, and $\de_j$, $j=1,2,3$, being complex parameters. Figure \ref{fig:fcc} shows a unit cell in $\Omega$.
\begin{figure}[H]
  \centering
  \includegraphics[width=0.4\textwidth]{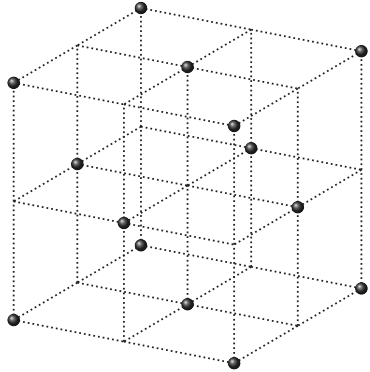}
  \caption{A unit cell of the $\Omega$ lattice.}
\label{fig:fcc}
  \end{figure}

Our study is motivated by two considerations. Firstly, the system \eqref{eqn:PDEs_P123456} satisfies the consistency around a cuboctahedron (CACO) property \cite{JN2019:arxiv1906.06650},
which is a generalization of the famous consistency around a cube (CAC) property\cite{NW2001:MR1869690}.
(See Appendix \ref{section:CACO_property} for a summary of the details of the CACO property and \S \ref{subsection:Background} for those of the CAC property.) Secondly, we are motivated by finding relations between partial difference equations and ordinary difference equations known as the discrete Painlev\'e equations. 

In this paper, we show that the system \eqref{eqn:PDEs_P123456} reduces to discrete Painlev\'e equations with initial value space characterised as $A_2^{(1)\ast}$ in the sense of Sakai\cite{SakaiH2001:MR1882403}. The latter equations have two forms in the literature given respectively by Tsuda \cite{TsudaT2008:MR2425662} and Ramani {\it et al.} \cite{RGH1991:MR1125951} and are explicitly given by:
\begin{subequations}\label{eqns:dP_E6_intro}
\begin{align}
 &\begin{cases}\label{eqn:asymmetric_dP_E6_intro}
 \Big(\underline{Y}+X\Big)\Big(X+Y\Big)
 =\dfrac{\left((X+c_3)^2-c_1\right)\left((X-c_3)^2-c_2\right)}{(X+t)^2-c_4},\\
 \Big(\overline{X}+Y)\Big(X+Y\Big)
 =\dfrac{\left((Y-c_3)^2-c_1\right)\left((Y+c_3)^2-c_2\right)}{\left(Y+t+\frac{1}{2}\right)^2-c_5},
 \end{cases}\\
 &\Big(\overline{X}+X\Big)\Big(\underline{X}+X\Big)=\dfrac{\left(X^2-c_1\right)\left(X^2-c_2\right)}{(X+t)^2-c_3}.
 \label{eqn:dP4_E6_intro}
\end{align}
\end{subequations}
Here, $t\in\bbC$ is an independent variable, $c_i$, $i=1\dots,5$, are complex parameters and
$X,Y$ are dependent variables:
\begin{equation}
 X=X(t),\quad
 Y=Y(t),\quad
 \overline{X}=X(t+1),\quad
 \underline{X}=X(t-1),\quad
 \underline{Y}=Y(t-1).
\end{equation}
We note that discrete Painlev\'e equations admit special solutions when parameters take special values.
For example, Equation \eqref{eqn:asymmetric_dP_E6_intro} has the special solution given by the generalized hypergeometric series ${}_3F_2$ when $4c_3+2\sqrt{c_4}+2\sqrt{c_5}=1$\cite{KajiwaraK2008:RIAMReports}.

Our main result is Theorem \ref{theorem:reduction}. To state the theorem, we first explain how to take the reduction on the lattice $\Omega$. To be explicit, consider a vertex $\bml\in \Omega$, given by $l_1\epsilon_1+l_2\epsilon_2+l_3\epsilon_3$. 
Define the plane $H_k\subset \Omega$ given by $l_3=k$.
We project the vertices of $H_1$ to the adjacent horizontal plane $H_0$ by taking 
$(l_1,l_2,1) \mapsto (l_1-1,l_2-1,0)$.
The union of the projection with the lattice points on $H_0$ forms $\mathbb Z^2$.
 We can define such a projection from every plane $H_k$ to $H_0$ by the following:
\[
(l_1,l_2,k) \mapsto (l_1-k,l_2-k,0).
\]
We call the result of this operation a {\it $(1,1,1)$-periodic reduction}. 

\begin{theorem}\label{theorem:reduction}
The $A_2^{(1)\ast}$-type discrete Painlev\'e equations \eqref{eqns:dP_E6_intro} can be obtained from 
the system of {\PDE}s \eqref{eqn:PDEs_P123456} via the $(1,1,1)$-periodic reduction.
\end{theorem}

\subsection{Notation and Definitions}\label{subsection:notation_definitions}
Throughout the paper, we use terminology to describe polynomials and quad-equations that is common in the literature. 
Readers who are unfamiliar with this notation may wish to consult \cite{JN2019:arxiv1906.06650,ABS2003:MR1962121,HJN2016:MR3587455}. 
We use $Q=Q(x,y,z,w)$ to denote a multivariable polynomial over $\bbC$. 
Under certain conditions, i.e., that $Q$ be affine linear and irreducible, we will refer to the equation $Q=0$ as a {\it quad-equation} or sometimes, for succinctness, refer to the polynomial $Q$ as a quad-equation. 
We remind the reader that the condition of irreducibility implies that $Q(x,y,z,w)=0$ can be solved for each argument, and that the solution is a rational function of the other three arguments.
\subsection{Background}
\label{subsection:Background}
Integrable systems are widely applicable models of science, occurring in fluid dynamics, particle physics and optics. 
The prototypical example is the famous Korteweg-de Vries (KdV) equation whose solitary wave-like solutions interact elastically like particles, leading to the invention of the term \textit{soliton}. 
It is then natural to ask what discrete versions of such equations are also integrable. 
This question turns out to be related to consistency conditions for polynomials associated to faces of cubes as we explain below.

Integrable discrete systems were discovered \cite{NQC1983:MR719638,NCWQ1984:MR763123,QNCL1984:MR761644,nimmo1998integrable} from mappings that turn out to be consistent on multi-dimensional cubes. (We note that there are additional systems that do not fall into this class, see e.g., \cite[Chapter 3]{HJN2016:MR3587455}.) 
These are quad-equations in the sense in \S \ref{subsection:notation_definitions}. 
In \cite{ABS2003:MR1962121,ABS2009:MR2503862,BollR2011:MR2846098,BollR2012:MR3010833}, 
Adler-Bobenko-Suris {\it et al.} classified quad-equations satisfying 
the consistency around a cube (CAC) property, which lead integrable {\PDE}s.
We refer to such \PDE s as ABS equations. 
It turns out that ABS equations contain many well known integrable \PDE s \cite{NC1995:MR1329559,NCWQ1984:MR763123,NQC1983:MR719638,HirotaR1977:MR0460934}.

Reductions of integrable PDEs lead to Painlev\'e equations, which first arose in the search for new transcendental functions in the early 1900's\cite{PainleveP1902:MR1554937,GambierB1910:MR1555055,FuchsR1905:quelques}. 
Again a natural question is to ask whether discrete versions exist with analogous properties. 
This question led to the discovery of second-order difference equations called the discrete Painlev\'e equations\cite{GR2004:MR2087743,KNY2017:MR3609039,QNCL1984:MR761644}).

It is now well-known that discrete Painlev\'e equations have initial value spaces with geometric structures that can be identified with root systems and affine Weyl groups \cite{SakaiH2001:MR1882403}. 
Sakai showed that there are 22 types of initial value spaces as shown in Table \ref{tab:sakai}.

\begin{table}[H]
\begin{center}
\caption{Types of spaces of initial values.}~\\
\label{tab:sakai}
\begin{tabular}{|l|l|}
\hline
Discrete type&Type of space of initial values\\
\hline
Elliptic&$A_0^{(1)}$\rule[-.5em]{0em}{1.6em}\\
\hline
Multiplicative&$A_0^{(1)\ast}$, $A_1^{(1)}$, $A_2^{(1)}$, $A_3^{(1)}$, \dots, $A_8^{(1)}$, $A_7^{(1)'}$\rule[-.5em]{0em}{1.6em}\\
\hline
Additive&$A_0^{(1)\ast\ast}$, $A_1^{(1)\ast}$, $A_2^{(1)\ast}$, $D_4^{(1)}$, \dots, $D_8^{(1)}$, $E_6^{(1)}$, $E_7^{(1)}$, $E_8^{(1)}$\rule[-.5em]{0em}{1.6em}\\
\hline
\end{tabular}
\end{center}
\end{table}

\subsection{Outline of the paper}
This paper is organized as follows.
In \S \ref{section:Weyl_dPE6}, we show the extended affine Weyl group of type $E_6^{(1)}$ and its subgroup which forms that of type $A_2^{(1)}$.
Moreover, from those birational actions we obtain the discrete Painlev\'e equations \eqref{eqns:dP_E6_intro} and the {\PDE}s \eqref{eqn:u_type123}, which are periodically reduced equations of the system \eqref{eqn:PDEs_P123456}.
In \S \ref{section:proof_theorem}, using the results in \S \ref{section:Weyl_dPE6} we give the proof of Theorem \ref{theorem:reduction}.
Finally, we give some concluding remarks in \S \ref{ConcludingRemarks}.
\section{Derivation of the discrete integrable systems from an extended affine Weyl group of type $E_6^{(1)}$}\label{section:Weyl_dPE6}
In this section, we derive the partial/ordinary discrete integrable systems from the birational actions of an extended affine Weyl group of type $E_6^{(1)}$, denoted by $\tW(E_6^{(1)})$.
Note that details of $\tW(E_6^{(1)})$ are given in Appendix \ref{section:Weyl_E6}.
\subsection{Extended affine Weyl group of type $A_2^{(1)}$}
Let $a_i$, $i=0,\dots,6$, be parameters satisfying the condition
\begin{equation}\label{eqn:para_a_E6}
 a_1+2a_2+3a_3+2a_4+a_5+2a_6+a_0=1,
\end{equation}
and $\tau^{(i)}_j$, $i=1,2,3$, $j=0,1,2,3$, be variables.
Moreover, we define the transformations $s_i$, $i=0,\dots,6$, $\iota_j$, $j=1,2,3$, by isomorphisms 
from the field of rational functions $K(\{\tau^{(i)}_j\})$, where $K=\mathbb{C}(\{a_i\})$, to itself.
These transformations collectively form the extended affine Weyl group of type $E_6^{(1)}$, denoted by $\tW(E_6^{(1)})$:
\begin{equation}\label{eqn:WE6}
 \tW(E_6^{(1)})=\langle s_0,\dots,s_6\rangle\rtimes\langle\iota_1,\iota_2,\iota_3\rangle.
\end{equation}
See Appendix \ref{section:Weyl_E6} for more details.

Let us define the transformations $w_i$, $i=0,1,2$, and $\pi$ by
\begin{equation}
 w_0=s_2s_1s_3s_2,\quad
 w_1=s_4s_5s_3s_4,\quad
 w_2=s_6s_0s_3s_6,\quad
 \pi=\iota_3\iota_1.
\end{equation}
They collectively form the extended affine Weyl group of type $A_2^{(1)}$:
\begin{equation}
 \tW(A_2^{(1)})=\langle w_0,w_1,w_2\rangle\rtimes\langle\pi\rangle.
\end{equation}
Indeed, the following fundamental relations hold:
\begin{equation}
 (w_iw_j)^{a_{ij}}=1,\quad i,j\in\{0,1,2\},\quad
 \pi^3=1,\quad
 \pi w_{\{0,1,2\}}=w_{\{1,2,0\}}\pi,
\end{equation}
where 
\begin{equation}
 (a_{ij})_{i,j=0}^2
 =\begin{pmatrix}
 2&3&3\\[-0.8em]
 3&2&3\\[-0.8em]
 3&3&2
 \end{pmatrix}.
\end{equation}
Introduce the parameters and variables that go well with $\tW(A_2^{(1)})$ as follows.
Let
\begin{subequations}
\begin{align}
 &b_0=a_1+2a_2+a_3,\quad
 b_1=a_3+2a_4+a_5,\quad
 b_2=a_3+2a_6+a_0,\\
 &c=\dfrac{a_0+a_1+a_3+a_5}{2},\quad
 d_{12}=\dfrac{a_0+a_1-a_3-a_5}{2},\quad
 d_{23}=\dfrac{a_0-a_1+a_3-a_5}{2},\\
 &d_{13}=\dfrac{a_0-a_1-a_3+a_5}{2},
\end{align}
\end{subequations}
where $b_0+b_1+b_2=1$, and
\begin{equation}
 y_1=\frac{\tau^{(1)}_1}{\tau^{(1)}_0},\quad
 y_2=\frac{\tau^{(3)}_3}{\tau^{(3)}_2},\quad
 y_3=\frac{\tau^{(2)}_1}{\tau^{(2)}_0},\quad
 y_4=\frac{\tau^{(2)}_3}{\tau^{(2)}_2},\quad
 y_5=\frac{\tau^{(1)}_3}{\tau^{(1)}_2},\quad
 y_6=\frac{\tau^{(3)}_1}{\tau^{(3)}_0}.
\end{equation}
Then, the action of $\tW(A_2^{(1)})$ on the parameters $b_0$, $b_1$, $b_2$, $c$, $d_{12}$, $d_{23}$, $d_{13}$ are given by
\begin{subequations}\label{eqns:WA2_paras}
\begin{align}
 &w_i(b_j)
 =\begin{cases}
 -b_i&\text{if }~ i=j,\\
 b_j+b_i&\text{if }~ i\neq j,
 \end{cases}\qquad
 w_0:(d_{12},d_{23})\mapsto(d_{23},d_{12}),\\
 &w_1:(d_{23},d_{13})\mapsto(d_{13},d_{23}),\qquad
 w_2:(d_{12},d_{13})\mapsto(-d_{13},-d_{12}),\\
 &\pi:(b_0,b_1,b_2,d_{12},d_{23},d_{13})\mapsto(b_1,b_2,b_0,-d_{23},-d_{13},d_{12}),
\end{align}
\end{subequations}
where $i,j\in\bbZ/(3\bbZ)$, while those on the $y$-variables $y_i$, $i=1,\dots,6$, are given by
\begin{subequations}\label{eqns:WA2_ys}
\begin{align}
 &w_0:
 \begin{pmatrix}
 y_1,~y_3\\
 y_5,~y_6
 \end{pmatrix}
 \mapsto
 \left(\begin{matrix}
 y_5,~\dfrac{(b_0-c+d_{13})y_3-(b_0-d_{12}+d_{23})y_1}{(b_0+d_{12}-d_{23})y_3-(b_0+c-d_{13})y_1}y_5\\[1.5em]
 y_1,~\dfrac{(b_0-c+d_{13})y_6-(b_0-d_{12}+d_{23})y_1}{(b_0+d_{12}-d_{23})y_6-(b_0+c-d_{13})y_1}y_5
 \end{matrix}\right),\\
 &w_1:
 \begin{pmatrix}
 y_1,~y_3\\
 y_4,~y_6
 \end{pmatrix}
 \mapsto
 \left(\begin{matrix}
 \dfrac{(b_1-c+d_{12})y_1-(b_1-d_{13}+d_{23})y_3}{(b_1+d_{13}-d_{23})y_1-(b_1+c-d_{12})y_3}y_4,~y_4\\[1.5em]
 y_3,~\dfrac{(b_1-c+d_{12})y_6-(b_1-d_{13}+d_{23})y_3}{(b_1+d_{13}-d_{23})y_6-(b_1+c-d_{12})y_3}y_4
 \end{matrix}\right),\\
 &w_2:
 \begin{pmatrix}
 y_1,~y_2\\
 y_3,~y_6
 \end{pmatrix}
 \mapsto
 \left(\begin{matrix}\dfrac{(b_2-c-d_{23})y_1-(b_2-d_{12}-d_{13})y_6}{(b_2+d_{12}+d_{13})y_1-(b_2+c+d_{23})y_6}y_2,~y_6\\[1.5em]
 \dfrac{(b_2-d_{12}-d_{13})y_6-(b_2-c-d_{23})y_3}{(b_2+c+d_{23})y_6-(b_2+d_{12}+d_{13})y_3}y_2,~y_2
 \end{matrix}\right),\\
 &\pi:(y_1,y_2,y_3,y_4,y_5,y_6)\mapsto
 (y_3,y_5,y_6,y_2,y_4,y_1).
\end{align}
\end{subequations}

\begin{remark}\label{remark:para_tau_action}
We follow the convention that the parameters and $y$-variables not explicitly included in the actions listed in Equations \eqref{eqns:WA2_paras} and \eqref{eqns:WA2_ys} are the ones that remain unchanged under the action of the corresponding transformation. That is, the transformation acts as an identity on those parameters or variables.
\end{remark}

For later convenience, we here define the translations in $\tW(A_2^{(1)})$ by
\begin{equation}
 T_1=w_1w_2\pi^2,\quad
 T_2=w_2w_0\pi^2,\quad
 T_3=w_0w_1\pi^2,
\end{equation}
whose actions on the parameters $b_0$, $b_1$, $b_2$, $c$, $d_{12}$, $d_{23}$, $d_{13}$ are given by
\begin{subequations}
\begin{align}
 T_1:&(b_0,b_1,d_{12},d_{13})\mapsto(b_0-1,b_1+1,-d_{12},-d_{13}),\\
 T_2:&(b_1,b_2,d_{12},d_{23})\mapsto(b_1-1,b_2+1,-d_{12},-d_{23}),\\
 T_3:&(b_2,b_0,d_{23},d_{13})\mapsto(b_2-1,b_0+1,-d_{23},-d_{13}).
\end{align}
\end{subequations}
Note that $T_1T_2T_3=1$ and $T_iT_j=T_jT_i$, where $i,j=1,2,3$, hold.
\subsection{Derivation of the partial difference equations from $\tW(A_2^{(1)})$}\label{subsection:E6_PDEs}
In this subsection, we derive the {\PDE}s \eqref{eqn:u_type123} from the birational action of $\tW(A_2^{(1)})$.

Let
\begin{equation}
 u_{l_1,l_2,l_3}={T_1}^{l_1}{T_2}^{l_2}{T_3}^{l_3}(y_2).
\end{equation}
Note that
\begin{equation}
 u_{0,1,1}=y_1,\quad
 u_{0,0,0}=y_2,\quad
 u_{1,0,0}=y_3,\quad
 u_{0,1,0}=y_4,\quad
 u_{1,1,0}=y_5,\quad
 u_{1,2,0}=y_6.
\end{equation}
We assign the variable $u_{l_1,l_2,l_3}$ on the vertices $(l_1,l_2,l_3)$ of the triangle lattice 
\begin{equation}
 \bbZ^3/(1,1,1):=\set{(l_1,l_2,l_3)\in\bbZ^3}{l_1+l_2+l_3=0}.
\end{equation}
Then, we obtain the following lemma.
\begin{lemma}
On the triangle lattice there are three fundamental relations (essentially two):
\begin{equation}\label{eqn:u_type123}
 \dfrac{u_{\ol{i}}}{u_{\ul{i}}}
 =\dfrac{\Big(b^{(i)}_{l_i,l_j}-c+(-1)^{l_i+l_j}d_{ij}\Big)u_{\ol{j}}-\Big(b^{(i)}_{l_i,l_j}-(-1)^{l_j+l_k}d_{jk}+(-1)^{l_i+l_k}d_{ik}\Big)u_{\ul{j}}}{\Big(b^{(i)}_{l_i,l_j}+(-1)^{l_j+l_k}d_{jk}-(-1)^{l_i+l_k}d_{ik}\Big)u_{\ol{j}}-\Big(b^{(i)}_{l_i,l_j}+c-(-1)^{l_i+l_j}d_{ij}\Big)u_{\ul{j}}},
\end{equation}
where $(i,j,k)=(1,2,3),\,(2,3,1),\,(3,1,2)$ and
\begin{equation}
 b^{(1)}_{l_1,l_2}=b_1+l_1-l_2,\quad
 b^{(2)}_{l_2,l_3}=b_2+l_2-l_3-1,\quad
 b^{(0)}_{l_1,l_3}=b_0+l_3-l_1.
\end{equation}
Here, $u=u_{l_1,l_2,l_3}$ and
the subscript \,$\bi$ (or, $\ul{i}$\,) for a function $u=u_{l_1,l_2,l_3}$
means \,$+1$ shift (or, $-1$ shift) in the $l_i$-direction.
\end{lemma}
\begin{proof}
Equations \eqref{eqn:u_type123} with $(i,j,k)=(1,2,3),\,(2,3,1),\,(3,1,2)$ are respectively obtained from the following actions:
\begin{subequations}\label{eqns:y_type123}
\begin{align}
 &\dfrac{T_1(y_5)}{y_4}
 =\dfrac{(b_1-c+d_{12}) y_6-(b_1+d_{23}-d_{13}) y_3}{(b_1-d_{23}+d_{13}) y_6-(b_1+c-d_{12}) y_3},
 \label{eqn:y_type1}\\
 &\dfrac{T_2(y_4)}{y_2}
 =\dfrac{(b_2-c-d_{23}) y_1-(b_2-d_{13}-d_{12}) y_6}{(b_2+d_{13}+d_{12}) y_1-(b_2+c+d_{23}) y_6},
 \label{eqn:y_type2}\\
 &\dfrac{T_3(y_2)}{y_5}
 =\dfrac{(b_0-c+d_{13}) y_3-(b_0-d_{12}+d_{23}) y_1}{(b_0+d_{12}-d_{23}) y_3-(b_0+c-d_{13}) y_1}.
 \label{eqn:y_type3}
\end{align}
\end{subequations}
Moreover, we can easily verify that
using Equations \eqref{eqn:u_type123}
we can express any $u_{l_1,l_2,l_3}$ on the lattice by the six initial variables $y_i$, $i=1,\dots,6$,
and one of the equations \eqref{eqn:u_type123} can be obtained from the other two equations.
Therefore, we have completed the proof.
\end{proof}

\begin{remark}
Because of the following relations:
\begin{subequations}
\begin{align}
 &w_0(u_{l_1,l_2,l_3})=u_{l_3,l_2,l_1}\,,\quad
 w_1(u_{l_1,l_2,l_3})=u_{l_2,l_1,l_3}\,,\quad
 w_2(u_{l_1,l_2,l_3})=u_{l_1+1,l_3+2,l_2}\,,\\
 &\pi(u_{l_1,l_2,l_3})=u_{l_3+1,l_1+1,l_2}\,,
\end{align}
\end{subequations}
which follow from
\begin{subequations}
\begin{align}
 &w_0T_{\{1,2,3\}}=T_{\{3,2,1\}}w_0,\quad
 w_1T_{\{1,2,3\}}=T_{\{2,1,3\}}w_1,\quad
 w_2T_{\{1,2,3\}}=T_{\{1,3,2\}}w_2,\\
 &\pi T_{\{1,2,3\}}=T_{\{2,3,1\}}\pi,\quad
 w_0(u_{0,0,0})=u_{0,0,0},\quad
 w_1(u_{0,0,0})=u_{0,0,0},\\
 &w_2(u_{0,0,0})=u_{1,2,0},\quad
 \pi(u_{0,0,0})=u_{1,1,0},
\end{align}
\end{subequations}
the transformation group $\tW(A_2^{(1)})$ can be also regarded as the symmetry of the triangle lattice (see Figure \ref{fig:fundamental_initial}).
\end{remark}

\begin{figure}[h]
\begin{center}
\includegraphics[width=0.7\textwidth]{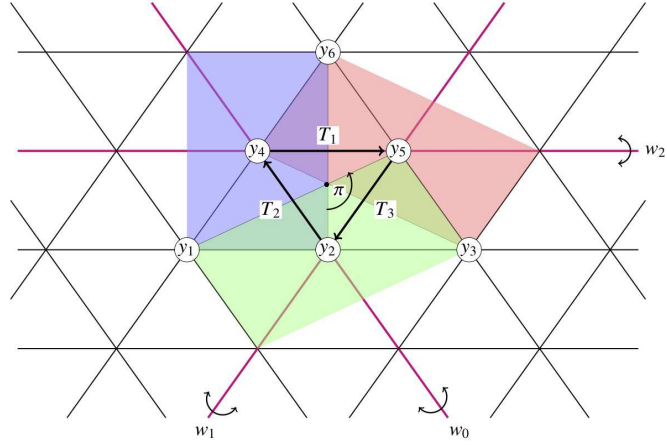}
\end{center}
\caption{
Triangle lattice. 
On the vertices the variables $u_{l_1,l_2,l_3}$ are assigned,
and on the quadrilaterals there exist quad-equations \eqref{eqn:u_type123}, e.g.
Equations \eqref{eqn:y_type1}, \eqref{eqn:y_type2} and \eqref{eqn:y_type3} are colored in red, blue and green, respectively.
}
\label{fig:fundamental_initial}
\end{figure}

\subsection{Derivation of the $A_2^{(1)\ast}$-type discrete Painlev\'e equations from $\tW(A_2^{(1)})$}\label{subsection:E6_ODEs}
In this subsection, we derive the $A_2^{(1)\ast}$-type discrete Painlev\'e equations \eqref{eqns:dP_E6_intro} from the birational action of $\tW(A_2^{(1)})$.

Let 
\begin{subequations}
\begin{align}
 &f=\dfrac{(c-d_{12}+d_{23}-d_{13})y_1}{2(y_6-y_1)}+\frac{b_2+c-d_{12}+d_{23}-d_{13}}{4},\\
 &g=\dfrac{(c-d_{12}+d_{23}-d_{13})y_3}{2(y_3-y_6)}-\frac{b_2+c-d_{12}+d_{23}-d_{13}}{4}.
\end{align}
\end{subequations}
Then, the action of $\tW(A_2^{(1)})$ on the variables $f$ and $g$ are given by
\begin{subequations}
\begin{align}
 &w_0(f)=f-\frac{b_0}{4},\quad
 w_1(g)=g+\frac{b_1}{4},\quad
 \pi(g)=f-\frac{b_2+b_0}{4},\\
 &\frac{4(c-d_{12}+d_{23}-d_{13})}{4 f-2 b_0-b_2-c-d_{12}+d_{23}+d_{13}}
 \left(w_0(g)+\frac{b_2+b_0+c+d_{12}-d_{23}-d_{13}}{4}\right)\notag\\
 &\hspace{1em}=\frac{(b_0-d_{12}+d_{23})(4g+b_2-c+d_{12}-d_{23}+d_{13})}{4f-b_2+c-d_{12}+d_{23}-d_{13}}\notag\\
 &\hspace{2.5em}-\frac{(b_0-c+d_{13})(4g+b_2+c-d_{12}+d_{23}-d_{13})}{4f-b_2-c+d_{12}-d_{23}+d_{13}},\\
 &\frac{4(c-d_{12}+d_{23}-d_{13})}{4g+2b_1+b_2+c-d_{12}-d_{23}+d_{13}}
 \left(w_1(f)-\frac{b_1+b_2+c-d_{12}-d_{23}+d_{13}}{4}\right)\notag\\
 &\hspace{1em}=\frac{(b_1+d_{23}-d_{13})(4f-b_2+c-d_{12}+d_{23}-d_{13})}{4g+b_2-c+d_{12}-d_{23}+d_{13}}\notag\\
 &\hspace{2.5em}-\frac{(b_1-c+d_{12})(4f-b_2-c+d_{12}-d_{23}+d_{13})}{4g+b_2+c-d_{12}+d_{23}-d_{13}},\\
 &\pi(f)=-\frac{(4f-b_2+c-d_{12}+d_{23}-d_{13})(4g+b_2+c-d_{12}+d_{23}-d_{13})}{16(f+g)}\notag\\
 &\hspace{3em}+\frac{b_0+c-d_{12}+d_{23}-d_{13}}{4}.
\end{align}
\end{subequations}
Using the transformation ${T_1}^2$ whose action on the parameter space $\{b_0,b_1,b_2,c,d_{12},d_{23},d_{13}\}$ is translational as ${T_1}^2:(b_0,b_1)\mapsto(b_0-2,b_1+2)$ shows,
we obtain the discrete Painlev\'e equation \eqref{eqn:asymmetric_dP_E6_intro}
with the following correspondence:
\begin{subequations}
\begin{align}
 &X=f,\quad
 Y=g,\quad
 \overline{X}={T_1}^2(f),\quad
 \underline{Y}={T_1}^{-2}(g),\quad
 t=\dfrac{2b_1+b_2-2}{4},\\
 &c_1=\dfrac{(b_2+c+d_{23})^2}{16},\quad
 c_2=\dfrac{(b_2-c-d_{23})^2}{16},\quad
 c_3=\dfrac{d_{12}+d_{13}}{4},\\
 &c_4=\dfrac{(c+d_{12}-d_{23}-d_{13})^2}{16},\quad
 c_5=\dfrac{(c-d_{12}-d_{23}+d_{13})^2}{16}.
\end{align}
\end{subequations}

We can also obtain the discrete Painlev\'e equations from non-translation on the parameter space as follows\cite{KNT2011:MR2773334}.
The action of $T_1$ on the parameter space:
\begin{equation*}
 T_1:(b_0,b_1,d_{12},d_{13})\mapsto(b_0-1,b_1+1,-d_{12},-d_{13}),
\end{equation*}
is not translational,
but when the parameters take the special values $d_{12}=d_{13}=0$,
it becomes translational motion on the parameter sub-space $\{b_0,b_1,b_2,c,d_{23}\}$: $T_1:(b_0,b_1)\mapsto(b_0-1,b_1+1)$.
Under the specialization of the parameters,
the action of $T_1$ gives the discrete Painlev\'e equation \eqref{eqn:dP4_E6_intro}
with the following correspondence:
\begin{subequations}
\begin{align}
 &X=2f,\quad
 \overline{X}=T_1(2f),\quad
 \underline{X}={T_1}^{-1}(2f),\quad
 t=\dfrac{2b_1+b_2-2}{2},\\
 &c_1=\dfrac{(b_2+c+d_{23})^2}{4},\quad
 c_2=\dfrac{(b_2-c-d_{23})^2}{4},\quad
 c_3=\dfrac{(c-d_{23})^2}{4}.
\end{align}
\end{subequations}

\section{Proof of Theorem \ref{theorem:reduction}}\label{section:proof_theorem}
In this section, we give the proof of Theorem \ref{theorem:reduction} via the reduction from the system of {\PDE}s \eqref{eqn:PDEs_P123456} to the system of {\PDE}s \eqref{eqn:u_type123}.

The following lemma holds.
\begin{lemma}\label{lemma:periodic_reduction}
By imposing the $(1,1,1)$-periodic condition: $u(\bml+\ep_1+\ep_2+\ep_3)=u(\bml)$
for $\bml\in\Omega$, 
the system \eqref{eqn:PDEs_P123456} can be reduced to the following system of {\PDE}s:
\begin{equation}\label{eqn:PDES_reduction}
 \dfrac{u_{\ol{i}}}{u_{\ul{i}}}
 =\dfrac{\Big(\al_{ij}-c+(-1)^{l_i+l_j}\de_i\Big)u_{\ol{j}}-\Big(\al_{ij}-(-1)^{l_j+l_k}\de_j+(-1)^{l_i+l_k}\de_k\Big)u_{\ul{j}}}
 {\Big(\al_{ij}+(-1)^{l_j+l_k}\de_j-(-1)^{l_i+l_k}\de_k\Big)u_{\ol{j}}-\Big(\al_{ij}+c-(-1)^{l_i+l_j}\de_i\Big)u_{\ul{j}}},
\end{equation}
where $(i,j,k)=(1,2,3),\,(2,3,1),\,(3,1,2)$, $u=u(\bml)$ and 
$\bml=\sum_{i=1}^3l_i\ep_i\in \bbZ^3/(\ep_1+\ep_2+\ep_3)$.
\end{lemma}
\begin{proof}
Applying the $(1,1,1)$-periodic condition to the system \eqref{eqn:PDEs_P123456},
we obtain Equations \eqref{eqn:PDES_reduction} with $(i,j,k)=(1,2,3)$, $(2,3,1)$ and $(3,1,2)$
from Equations \eqref{eqn:PDEs_P123456} with $(i,j,k)=(1,2,3)$, $(2,3,1)$ and $(3,1,2)$, respectively.
Therefore, we have completed the proof.
\end{proof}

\begin{remark}
\quad\\[-1.5em]
\begin{description}
\item[(i)]
The number of essential equations in the system \eqref{eqn:PDES_reduction} is two.
\item[(ii)]
By the $(1,1,1)$-reduction,
each cuboctahedron is reduced to a hexagram {\rm(see Figure \ref{fig:reduction_cuboctahedron})},
which causes the reduction from the face-centred cubic lattice $\Omega$ to the triangle lattice $\bbZ^3/(\ep_1+\ep_2+\ep_3)$.
\end{description}
\end{remark}

\begin{figure}[htbp]
\begin{center}
\includegraphics[width=1\textwidth]{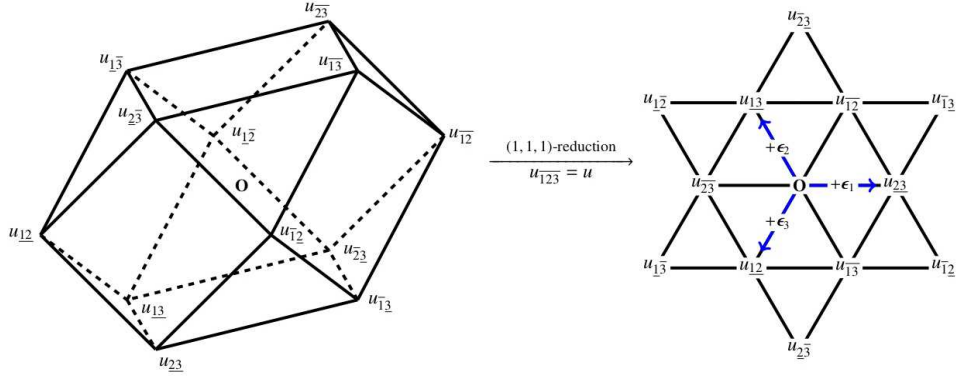}
\end{center}
\caption{The $(1,1,1)$-reduction of the cuboctahedron.}
\label{fig:reduction_cuboctahedron}
\end{figure}

\begin{lemma}\label{lemma:correspondence_tau_CACO}
The reduced system \eqref{eqn:PDES_reduction} is equivalent to equations the system \eqref{eqn:u_type123}.
\end{lemma}
\begin{proof}
The statement follows from the following correspondences:
\begin{subequations}
\begin{align}
 &b^{(1)}_{l_1,l_2}=\al_{12},\quad
 b^{(2)}_{l_2,l_3}=\al_{23},\quad
 b^{(0)}_{l_1,l_3}=\al_{31},\quad
 d_{12}=\de_1,\quad
 d_{23}=\de_2,\quad
 d_{13}=\de_3,\\
 &u_{l_1,l_2,l_3}=u(l_1\ep_1+l_2\ep_2+l_3\ep_3).
\end{align}
\end{subequations}
\end{proof}

\begin{remark}
Lemma \ref{lemma:correspondence_tau_CACO} means that the reduced system \eqref{eqn:PDES_reduction} can be obtained from the theory of the $\tau$-function associated with $A_2^{(1)\ast}$-type discrete Painlev\'e equations.
\end{remark}

We are now ready to prove Theorem \ref{theorem:reduction}.
The $(1,1,1)$-periodic reduction from the system \eqref{eqn:PDEs_P123456} to the system \eqref{eqn:PDES_reduction} given in Lemma \ref{lemma:periodic_reduction},
the relation between the system \eqref{eqn:PDES_reduction} and the system \eqref{eqn:u_type123} given in Lemma \ref{lemma:correspondence_tau_CACO},
and that between the system \eqref{eqn:u_type123} and the $A_2^{(1)\ast}$-type discrete Painlev\'e equations \eqref{eqns:dP_E6_intro} 
given in \S \ref{subsection:E6_PDEs} and \S \ref{subsection:E6_ODEs}
collectively give the proof of Theorem \ref{theorem:reduction}.
\section{Concluding remarks}\label{ConcludingRemarks}
In this paper, we considered a reduction of a system of {\PDE}s, which is unusual in the sense that it has the CACO property but not the widely studied CAC property. We showed how the system \eqref{eqn:PDEs_P123456} can be reduced to the $A_2^{(1)\ast}$-type discrete Painlev\'e equations \eqref{eqns:dP_E6_intro} using the affine Weyl group associated with the discrete Painlev\'e equations. 

In a forthcoming paper (N. Joshi and N. Nakazono), we will show how another system of {\PDE}s, which also has the CACO property, can be reduced to the $A_2^{(1)}$-type discrete Painlev\'e equations (see Table \ref{tab:sakai} for the distinction between $A_2^{(1)}$ and $A_2^{(1)\ast}$).
\subsection*{Acknowledgment}
N. Nakazono would like to thank Profs M. Noumi, Y. Ohta and Y. Yamada for inspiring and fruitful discussions.
This research was supported by an Australian Laureate Fellowship \# FL120100094 and grant \# DP160101728 from the Australian Research Council and JSPS KAKENHI Grant Numbers JP19K14559 and JP17J00092.
\appendix
\section{Consistency around a cuboctahedron property}\label{section:CACO_property}
In this appendix, we recall the definition of {\it consistency around a cuboctahedron}.
To define it, we also introduce an additional important property called {\it consistency around an octahedron}.
We refer the reader to \cite{JN2019:arxiv1906.06650} for detailed information about these properties.

\subsection{Consistency around an octahedron property}\label{subsection:def_CAO}
In this subsection, we give a definition of a consistency around an octahedron.

Let $u_i$, $i=1,\dots,6$, be variables and consider the octahedron shown in Figure \ref{fig:octahedron_3D}.
The planes that pass through the vertices $\{u_4,u_2,u_1,u_5\}$, $\{u_2,u_6,u_5,u_3\}$ and $\{u_6,u_4,u_3,u_1\}$ give 3 quadrilaterals that lie in the interior of the octahedron and we assign the quad-equations $Q_i$, $i=1,2,3$, to the quadrilaterals as the following:
\begin{equation}\label{eqn:octahedron_Q1Q2Q3}
 Q_1\left(u_4,u_2,u_1,u_5\right)=0,\quad
 Q_2\left(u_2,u_6,u_5,u_3\right)=0,\quad
 Q_3\left(u_6,u_4,u_3,u_1\right)=0.
\end{equation}
The consistency around an octahedron property is defined by the following.

\begin{figure}[htbp]
\begin{center}
 \includegraphics[width=0.3\textwidth]{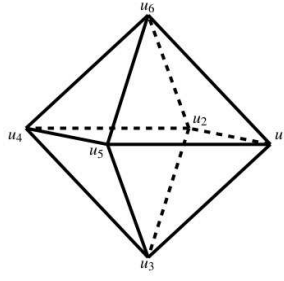}
\end{center}
\caption{An octahedron labelled with vertices $u_i$, $i=1,\dots,6$.}
\label{fig:octahedron_3D}
\end{figure}

\begin{definition}[CAO property\cite{JN2019:arxiv1906.06650}]\rm
The octahedron with quad-equations $\{Q_1,Q_2,Q_3\}$ is said to have a {\it consistency around an octahedron (CAO) property}, 
if each quad-equation can be obtained from the other two equations.
An octahedron is said to be a {\it CAO octahedron}, if it has the CAO property.
\end{definition}

\subsection{Consistency around a cuboctahedron property}\label{subsection:def_CACO}
In this subsection, we give a definition of a consistency around a cuboctahedron.

We consider the cuboctahedron centered around the origin whose twelve vertices are given by
$V=\set{\pm\ep_i\pm\ep_j}{i,j\in\bbZ,~1\leq i<j\leq 3}$,
where $\{\ep_1,\ep_2,\ep_3\}$ form the standard basis of $\bbR^3$.
We assign the variables $u(\bml)$ to the vertices $\bml\in V$
and impose the following relations:
\begin{subequations}\label{eqns:V_caco_general}
\begin{align}
 &Q_1\left(u_5,u_1,v_5,v_4\right)=0,\quad
 Q_2\left(v_2,v_1,u_2,u_4\right)=0,\quad
 Q_3\left(u_3,u_5,v_3,v_2\right)=0,
 \label{eqn:V_caco_general_1}\\
 &Q_4\left(v_6,v_5,u_6,u_2\right)=0,\quad
 Q_5\left(u_1,u_3,v_1,v_6\right)=0,\quad
 Q_6\left(v_4,v_3,u_4,u_6\right)=0,
 \label{eqn:V_caco_general_2}\\
 &Q_7\left(u_4,u_2,u_1,u_5\right)=0,\quad
 Q_8\left(u_2,u_6,u_5,u_3\right)=0,\quad
 Q_9\left(u_6,u_4,u_3,u_1\right)=0,
 \label{eqn:V_caco_general_3}
\end{align}
\end{subequations}
where $Q_i$, $i=1,\dots,9$, are quad-equations and
\begin{subequations}
\begin{align}
 &u_1=u(\ep_2+\ep_3),
 &&u_2=u(-\ep_1-\ep_3),
 &&u_3=u(\ep_1+\ep_2),
 &&u_4=u(-\ep_2-\ep_3),\\
 &u_5=u(\ep_1+\ep_3),
 &&u_6=u(-\ep_1-\ep_2),
 &&v_1=u(\ep_2-\ep_3),
 &&v_2=u(\ep_1-\ep_3),\\
 &v_3=u(\ep_1-\ep_2),
 &&v_4=u(-\ep_2+\ep_3),
 &&v_5=u(-\ep_1+\ep_3),
 &&v_6=u(-\ep_1+\ep_2).
\end{align}
\end{subequations}
Note that quad-equations $Q_i$, $i=1,\dots,6$, are assigned to the faces of the cuboctahedron (see Figure \ref{fig:cuboctahedron_3D_a}).
Moreover, $u_i$, $i=1,\dots,6$, collectively form the vertices of an octahedron
and quad-equations $Q_i$, $i=7,8,9$, are assigned to the quadrilaterals that appear as sections passing through four vertices of the octahedron
(see Figure \ref{fig:cuboctahedron_3D_b}).

\begin{figure}[htbp]
\centering
\begin{subfigure}[b]{0.43\textwidth}
\centering
 \includegraphics[width=\textwidth]{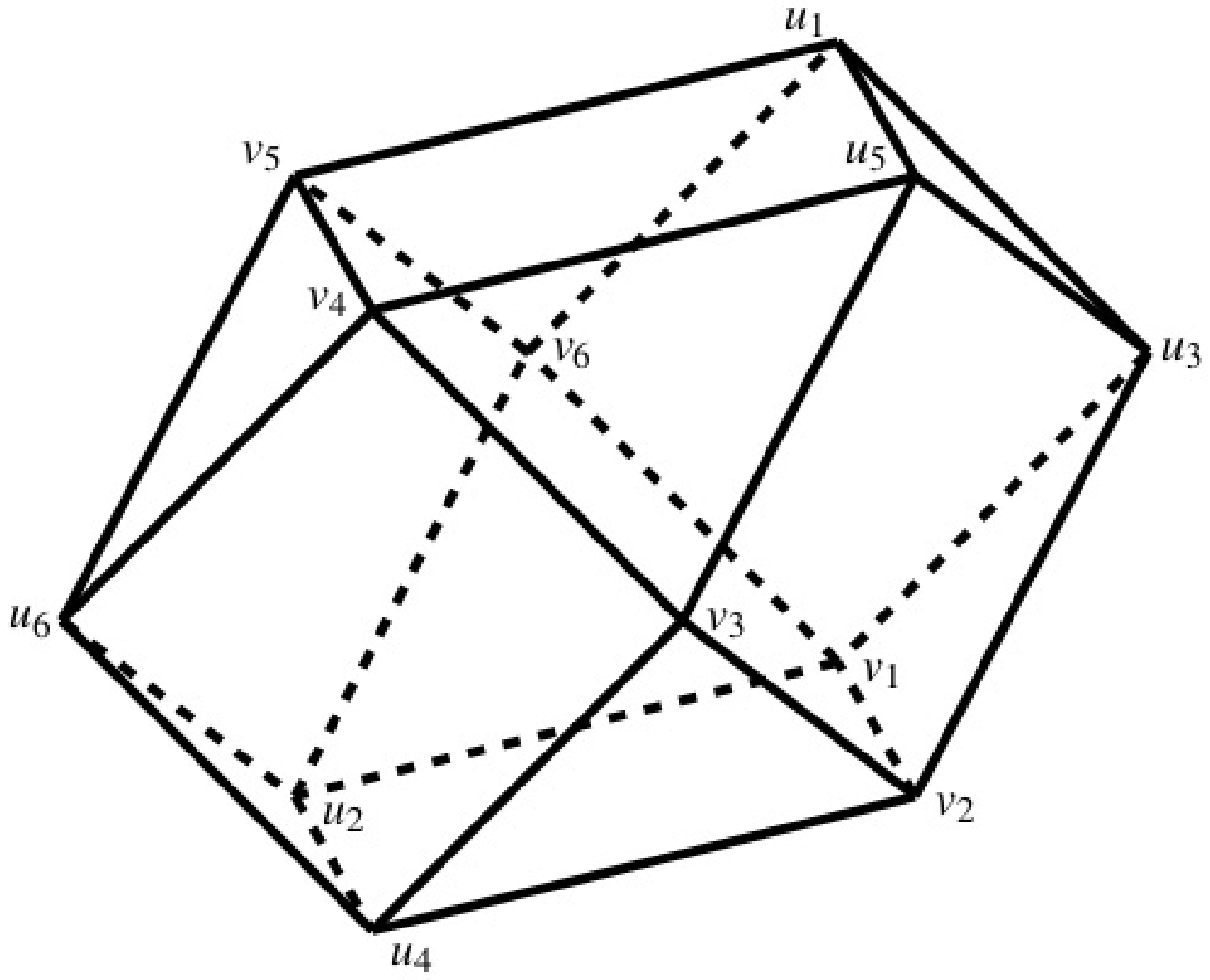}
 \caption{A cuboctahedron labelled with vertices $u_i$ and $v_j$, $i,j=1,\dots,6$.}
 \label{fig:cuboctahedron_3D_a}
\end{subfigure}
\hspace{2em}
\begin{subfigure}[b]{0.43\textwidth}
\centering
 \includegraphics[width=\textwidth]{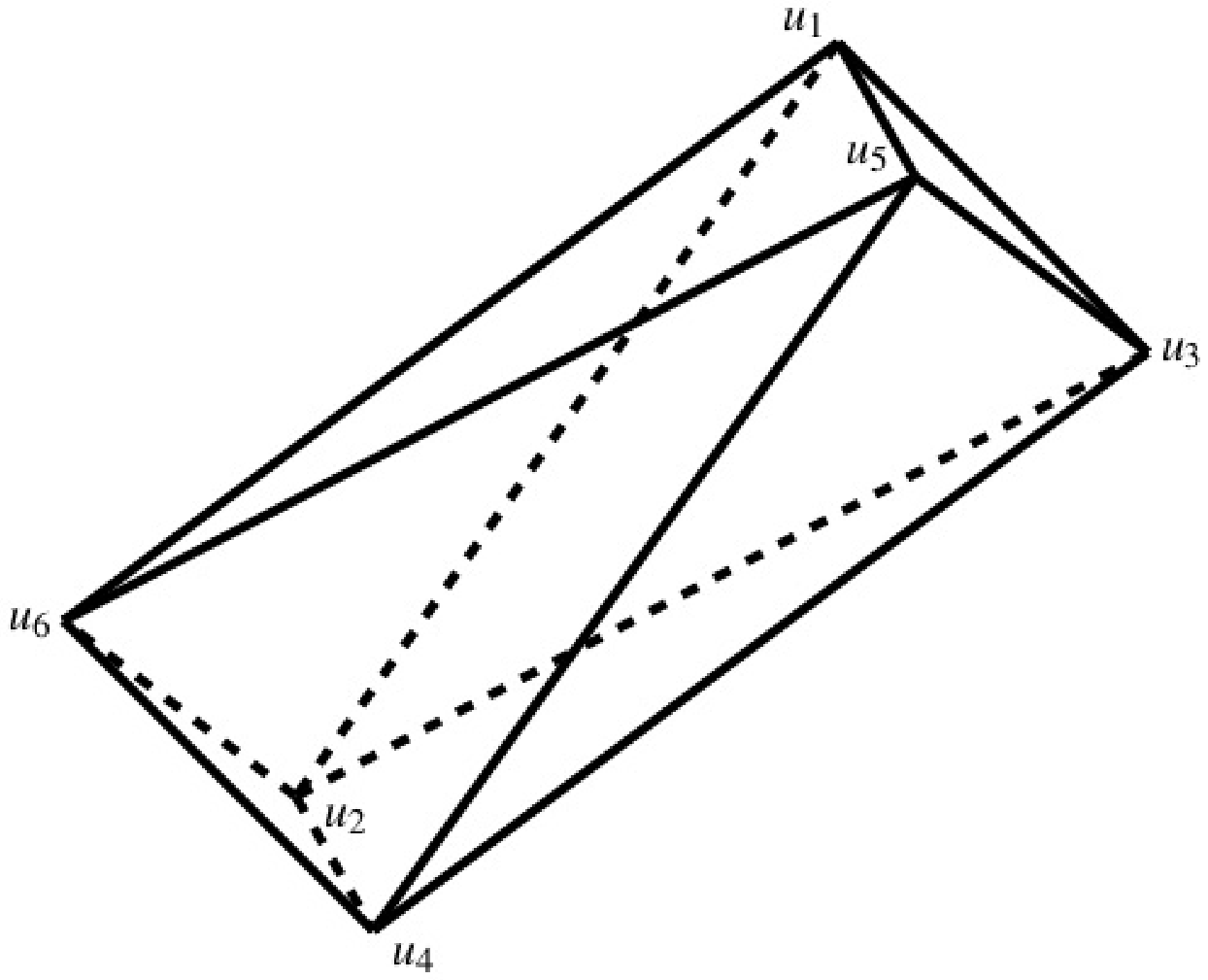}
 \caption{An octahedron labelled with vertices $u_i$, $i=1,\dots,6$.}
 \label{fig:cuboctahedron_3D_b}
\end{subfigure}
 \caption{A cuboctahedron and an interior octahedron.}
 \label{fig:cuboctahedron_3D}
\end{figure}

We are now in a position to give the following definitions.
\begin{definition}[CACO property\cite{JN2019:arxiv1906.06650}]\label{def:CACO_cuboctahedron}\rm
The cuboctahedron with quad-equations $\{Q_1,\dots,Q_9\}$ is said to have a {\it consistency around a cuboctahedron (CACO) property}, 
if the following properties hold.
\begin{description}
\item[(i)] 
The octahedron with quad-equations $\{Q_7,Q_8,Q_9\}$ has the CAO property.
\item[(ii)]
Assume that $u_1,\dots,u_6$ are given so as to satisfy $Q_i=0$, $i=7,8,9$, and, in addition, $v_k$ is given, for some $k\in\{1,\dots,6\}$.
Then, quad-equations $Q_i$, $i=1\dots,6$, determine the variables $v_j$, $j\in\{1,\dots,6\}\backslash\{k\}$, uniquely.
\end{description}
A cuboctahedron is said to be a {\it CACO cuboctahedron}, if it has the CACO property.
\end{definition}

\begin{definition}[Square property\cite{JN2019:arxiv1906.06650}]\rm
The CACO cuboctahedron with quad-equations $\{Q_1,\dots,Q_9\}$ is said to have a {\it square property}, 
if there exist polynomials $K_i=K_i(x,y,z,w)$, $i=1,2,3$, where
$\deg_x{K_i}=\deg_w{K_i}=1$ and $1\leq\deg_y{K_i},\,\deg_z{K_i}$,
satisfying 
\begin{equation}
 K_1(v_1,u_1,u_4,v_4)=0,\quad 
 K_2(v_2,u_2,u_5,v_5)=0,\quad
 K_3(v_3,u_3,u_6,v_6)=0.
\end{equation}
Then, each equation $K_i=0$ is called a {\it square equation}.
\end{definition}

\subsection{CACO property of \PDE s}
We now explain how to associate quad-equations with \PDE s in three-dimensional space by using the system of {\PDE}s \eqref{eqn:PDEs_P123456} as an example. 
This requires us to consider overlapping cuboctahedra that lead to two-dimensional tessellations consisting of quadrilaterals. 
For each given cuboctahedron, there are twelve overlapping cuboctahedra. 

The twelve overlapping cuboctahedra around a given one provide six directions of tiling by quadrilaterals. For later convenience, we label directions by $\ep_i\pm \ep_j$, $1\leq i<j\leq 3$. Vertices labelled in this way form the set $\Omega$ given by \eqref{eqn:def_Omega}.
Such vertices are interpreted as being iterated on each successive cuboctahedron. 
We here consider the system of {\PDE}s \eqref{eqn:PDEs_P123456}.
For simplicity, we abbreviate each respective equation in Equations \eqref{eqn:PDEs_P123456} as 
\begin{subequations}
\begin{align}
 &P_1\left(u_{\ol{13}},u_{\ol{23}},u_{\ul{1}\ol{3}},u_{\ul{2}\ol{3}}\right)=0,~
 P_2\left(u_{\ol{12}},u_{\ol{13}},u_{\ol{1}\ul{2}},u_{\ol{1}\ul{3}}\right)=0,~
 P_3\left(u_{\ol{23}},u_{\ol{12}},u_{\ol{2}\ul{3}},u_{\ul{1}\ol{2}}\right)=0,\\
 &P_4\left(u_{\ul{23}},u_{\ul{13}},u_{\ol{23}},u_{\ol{13}}\right)=0,~
 P_5\left(u_{\ul{13}},u_{\ul{12}},u_{\ol{13}},u_{\ol{12}}\right)=0,~
 P_6\left(u_{\ul{12}},u_{\ul{23}},u_{\ol{12}},u_{\ol{23}}\right)=0.
\end{align}
\end{subequations}

Conversely, given $\bml\in\Omega$, we obtain the cuboctahedron centered around $\bml$. We refer to its quad-equations as before by
$\{Q_1(\bml),\dots,Q_9(\bml)\}$. Moreover, the overlapped region gives an octahedron centred around $\bml+\ep_3$, and we label its quad-equations by $\{\hat{Q}_1(\bml),\hat{Q}_2(\bml),\hat{Q}_3(\bml)\}$.

Each such quad-equation is identified with the 6 partial difference equations given in Equations \eqref{eqn:PDEs_P123456} in the following way.
For $Q_1,\ldots, Q_9$, we use
\begin{subequations}
\begin{align}
 &Q_1(\bml)=P_1\left(u_{\ol{13}},u_{\ol{23}},u_{\ul{1}\ol{3}},u_{\ul{2}\ol{3}}\right)=0,\quad
 Q_2(\bml)=P_1\left(u_{\ol{1}\ul{3}},u_{\ol{2}\ul{3}},u_{\ul{13}},u_{\ul{23}}\right)=0,\\
 &Q_3(\bml)=P_2\left(u_{\ol{12}},u_{\ol{13}},u_{\ol{1}\ul{2}},u_{\ol{1}\ul{3}}\right)=0,\quad
 Q_4(\bml)=P_2\left(u_{\ul{1}\ol{2}},u_{\ul{1}\ol{3}},u_{\ul{12}},u_{\ul{13}}\right)=0,\\
 &Q_5(\bml)=P_3\left(u_{\ol{23}},u_{\ol{12}},u_{\ol{2}\ul{3}},u_{\ul{1}\ol{2}}\right)=0,\quad
 Q_6(\bml)=P_3\left(u_{\ul{2}\ol{3}},u_{\ol{1}\ul{2}},u_{\ul{23}},u_{\ul{12}}\right)=0,\\
 &Q_7(\bml)=P_4\left(u_{\ul{23}},u_{\ul{13}},u_{\ol{23}},u_{\ol{13}}\right)=0,\quad
 Q_8(\bml)=P_5\left(u_{\ul{13}},u_{\ul{12}},u_{\ol{13}},u_{\ol{12}}\right)=0,\\
 &Q_9(\bml)=P_6\left(u_{\ul{12}},u_{\ul{23}},u_{\ol{12}},u_{\ol{23}}\right)=0,
\end{align}
\end{subequations}
and for $\hat{Q}_1,\hat{Q}_2,\hat{Q}_3$, we use
\begin{subequations}
\begin{align}
 &\hat{Q}_1(\bml)=P_1\left(u_{\ol{13}},u_{\ol{23}},u_{\ul{1}\ol{3}},u_{\ul{2}\ol{3}}\right)=0,\quad
 \hat{Q}_2(\bml)=P_2\left(u_{\ol{23}},u_{\ol{33}},u_{\ul{2}\ol{3}},u\right)=0,\\
 &\hat{Q}_3(\bml)=P_3\left(u_{\ol{33}},u_{\ol{13}},u,u_{\ul{1}\ol{3}}\right)=0. 
\end{align}
\end{subequations}
Then, the following proposition holds.
\begin{proposition}[\cite{JN2019:arxiv1906.06650}]
The system of {\PDE}s \eqref{eqn:PDEs_P123456} has the CACO and square properties, 
that is, the following statements hold.
\begin{description}
\item[(i)]
The cuboctahedra with quad-equations $\{Q_i(\bml)\}$ have the CACO and square properties.
\item[(ii)]
The square equations are consistent with the {\PDE}s \eqref{eqn:PDEs_P123456}.
\item[(iii)]
The octahedra with quad-equations $\{\hat{Q}_i(\bml)\}$ have the CAO property.
\end{description}
\end{proposition}

\section{Extended affine Weyl group of type $E_6^{(1)}$ and $\tau$-variables}\label{section:Weyl_E6}
In this appendix, we review the action of the extended affine Weyl group of type $E_6^{(1)}$ given in \cite{TsudaT2008:MR2425662}, which is the symmetry group of $A_2^{(1)\ast}$-type discrete Painlev\'e equations.

Let $a_i$, $i=0,\dots,6$, be parameters satisfying the condition \eqref{eqn:para_a_E6}
and $\tau^{(i)}_j$, $i=1,2,3$, $j=0,1,2,3$, be variables.
The actions of transformations $s_i$, $i=0,\dots,6$, and $\iota_j$, $j=1,2,3$, 
on the parameters are given by
\begin{subequations}\label{eqns:WE6_paras}
\begin{align}
 &s_0:(a_0,a_6)\mapsto(-a_0,a_6+a_0),\qquad
 s_1:(a_1,a_2)\mapsto(-a_1,a_2+a_1),\\
 &s_2:(a_1,a_2,a_3)\mapsto(a_1+a_2,-a_2,a_3+a_2),\\
 &s_3:(a_2,a_3,a_4,a_6)\mapsto(a_2+a_3,-a_3,a_4+a_3,a_6+a_3),\\
 &s_4:(a_3,a_4,a_5)\mapsto(a_3+a_4,-a_4,a_5+a_4),\qquad
 s_5:(a_4,a_5)\mapsto(a_4+a_5,-a_5),\\
 &s_6:(a_0,a_3,a_6)\mapsto(a_0+a_6,a_3+a_6,-a_6),\qquad\\
 &\iota_1 a_{\{0,5,4,6\}}\mapsto a_{\{5,0,4,6\}},\quad
 \iota_2 a_{\{0,1,2,6\}}\mapsto a_{\{1,0,6,2\}},\quad
 \iota_3 a_{\{1,5,2,4\}}\mapsto a_{\{5,1,4,2\}},
\end{align}
\end{subequations}
while those on the $\tau$-variables $\tau^{(i)}_j$, $i=1,2,3$, $j=0,1,2,3$, are given by
\begin{subequations}\label{eqns:WE6_tau}
\begin{align}
 s_0&:(\tau^{(3)}_2,\tau^{(3)}_3)\mapsto(\tau^{(3)}_3,\tau^{(3)}_2),\qquad
 s_1:(\tau^{(1)}_2,\tau^{(1)}_3)\mapsto(\tau^{(1)}_3,\tau^{(1)}_2),\\
 s_2&:(\tau^{(1)}_1,\tau^{(1)}_2)\mapsto(\tau^{(1)}_2,\tau^{(1)}_1),\\
 &:(\tau^{(2)}_0,\tau^{(3)}_0)\mapsto
 \left(\frac{(a_2+a_3) \tau^{(1)}_1 \tau^{(2)}_0-a_2 \tau^{(2)}_1 \tau^{(1)}_0}{a_3 \tau^{(1)}_2},
 \frac{(a_2+a_3) \tau^{(1)}_1 \tau^{(3)}_0-a_2 \tau^{(3)}_1 \tau^{(1)}_0}{a_3 \tau^{(1)}_2}\right),\\
 s_3&:(\tau^{(1)}_1,\tau^{(2)}_1,\tau^{(3)}_1,\tau^{(1)}_0,\tau^{(2)}_0,\tau^{(3)}_0)\mapsto
 (\tau^{(1)}_0,\tau^{(2)}_0,\tau^{(3)}_0,\tau^{(1)}_1,\tau^{(2)}_1,\tau^{(3)}_1),\\
 s_4&:(\tau^{(2)}_1,\tau^{(2)}_2)\mapsto(\tau^{(2)}_2,\tau^{(2)}_1),\\
 &:(\tau^{(1)}_0,\tau^{(3)}_0)\mapsto
 \left(\frac{(a_3+a_4) \tau^{(2)}_1 \tau^{(1)}_0-a_4 \tau^{(1)}_1 \tau^{(2)}_0}{a_3 \tau^{(2)}_2},
 \frac{(a_3+a_4) \tau^{(2)}_1 \tau^{(3)}_0-a_4 \tau^{(3)}_1 \tau^{(2)}_0}{a_3 \tau^{(2)}_2}\right),\\
 s_5&:(\tau^{(2)}_2,\tau^{(2)}_3)\mapsto(\tau^{(2)}_3,\tau^{(2)}_2),\\
 s_6&:(\tau^{(3)}_1,\tau^{(3)}_2)\mapsto(\tau^{(3)}_2,\tau^{(3)}_1),\\
 &:(\tau^{(1)}_0,\tau^{(2)}_0)\mapsto
 \left(\frac{(a_3+a_6) \tau^{(3)}_1 \tau^{(1)}_0-a_6 \tau^{(1)}_1 \tau^{(3)}_0}{a_3 \tau^{(3)}_2},
 \frac{(a_3+a_6) \tau^{(3)}_1 \tau^{(2)}_0-a_6 \tau^{(2)}_1 \tau^{(3)}_0}{a_3 \tau^{(3)}_2}\right),\\
 \iota_1&:(\tau^{(2)}_j,\tau^{(3)}_j)\mapsto(\tau^{(3)}_j,\tau^{(2)}_j),\qquad
 \iota_2:(\tau^{(1)}_j,\tau^{(3)}_j)\mapsto(\tau^{(3)}_j,\tau^{(1)}_j),\\
 \iota_3&:(\tau^{(1)}_j,\tau^{(2)}_j)\mapsto(\tau^{(2)}_j,\tau^{(1)}_j),\qquad
 j=0,1,2,3.
\end{align}
\end{subequations}

\begin{remark}
\quad\\[-1.2em]
\begin{description}
\item[(i)]
Each transformation here defined is an isomorphism from the field of rational functions $K(\{\tau^{(i)}_j\})$, where $K=\mathbb{C}(\{a_i\})$, to itself.
\item[(ii)] 
We follow the convention of Remark \ref{remark:para_tau_action}, for Equations \eqref{eqns:WE6_paras} and \eqref{eqns:WE6_tau}. That is, each transformation acts as an identity on  parameters or variables not appearing in its definition.
\end{description}
\end{remark}

The transformations collectively form the extended affine Weyl group of type $E_6^{(1)}$, denoted by \eqref{eqn:WE6}.
Indeed, the following fundamental relations hold:
\begin{subequations}\label{eqns:fundamental_E6}
\begin{align}
 &(s_is_j)^{A_{ij}}=1,\quad
 {\iota_1}^2={\iota_2}^2={\iota_3}^2=1,\quad
 \iota_1\iota_2=\iota_2\iota_3=\iota_3\iota_1,\quad
 \iota_2\iota_1=\iota_3\iota_2=\iota_1\iota_3,\\
 &\iota_1 s_{\{0,5,4,6\}}= s_{\{5,0,4,6\}}\iota_1,\quad
 \iota_2 s_{\{0,1,2,6\}}= s_{\{1,0,6,2\}}\iota_2,\quad
 \iota_3 s_{\{1,5,2,4\}}= s_{\{5,1,4,2\}}\iota_3,
\end{align}
\end{subequations}
where $i,j\in\{0,1,\dots,6\}$ and
\begin{equation}
 (A_{ij})_{i,j=0}^6
 =\begin{pmatrix}
 2&0&0&0&0&0&3\\[-0.8em]
 0&2&3&0&0&0&0\\[-0.8em]
 0&3&2&3&0&0&0\\[-0.8em]
 0&0&3&2&3&0&3\\[-0.8em]
 0&0&0&3&2&3&0\\[-0.8em]
 0&0&0&0&3&2&0\\[-0.8em]
 3&0&0&3&0&0&2
\end{pmatrix}.
\end{equation}

\begin{remark}
The correspondence between the notations in this paper and those in \cite{TsudaT2008:MR2425662} is given by
$\tau^{(i)}_j\to \tau^{i}_j$ and 
$\tau^{(i)}_0\to s_3(\tau^i_1)$, where $i,j=1,2,3$.
\end{remark}

\def\cprime{$'$} \def\cprime{$'$}

\end{document}